\def\draft{0}
\documentclass[11 pt]{article}

\usepackage{comment}

\usepackage{amsmath,amsfonts,amssymb}
\newtheorem{theorem}{Theorem}[section]
\newtheorem{nt}{Selfnote}

\newtheorem{definition}[theorem]{Definition}
\newtheorem{lemma}[theorem]{Lemma}

\newtheorem{corollary}[theorem]{Corollary}
\newtheorem{claim}[theorem]{Claim}
\newtheorem{fact}[theorem]{Fact}

\newtheorem{remk}[theorem]{Remark}
\newtheorem{exmp}[theorem]{Example}

\def\FullBox{\hbox{\vrule width 8pt height 8pt depth 0pt}}

\def\qed{\ifmmode\qquad\FullBox\else{\unskip\nobreak\hfil
\penalty50\hskip1em\null\nobreak\hfil\FullBox
\parfillskip=0pt\finalhyphendemerits=0\endgraf}\fi}

\def\qedsketch{\ifmmode\Box\else{\unskip\nobreak\hfil
\penalty50\hskip1em\null\nobreak\hfil$\Box$
\parfillskip=0pt\finalhyphendemerits=0\endgraf}\fi}

\newenvironment{proof}{\begin{trivlist} \item {\bf Proof:~~}}
  {\qed\end{trivlist}}

\newcommand{\etal}{{\it et~al.\ }}
\newcommand{\ie} {{\it i.e.,\ }}

\newfont{\inhead}{eufm10 scaled\magstep1}

\def\B{\{0,1\}}

\renewcommand{\Pr}{\mathop{\mathbb P}\displaylimits}

\newcommand{\eps}{\varepsilon}

\newcommand{\dd}{{n}}

\newcommand{\mcal}[1]{\mathcal{#1}}

\ifnum\draft=1
\newcommand{\anote}[1]{[{\bf Anindya's Note: #1]}}
\else
\newcommand{\anote}[1]{}
\fi
\begin{document}

\title{Lower bounds in differential privacy \ifnum\draft=1 \\ 
\textsc{\small Working Draft: Please Do Not Distribute}\fi}
\author{Anindya De\thanks{Computer Science Division, University of California, Berkeley.  Most of this work was done while the author was a summer intern at Microsoft Research, Silicon Valley. This material is based upon work supported by Luca Trevisan's National Science Foundation grant No.  CCF-1017403. \tt{anindya@cs.berkeley.edu} \tt} }
%
\begin{titlepage}

\maketitle

\begin{abstract}
This is a paper about private data analysis, in which a trusted curator
holding a confidential database responds  to real vector-valued queries.
A common approach to ensuring privacy for the database elements is to
add appropriately generated random noise to the answers, releasing only
these {\em noisy} responses.
A line of study initiated in~\cite{DN03} examines the
amount of distortion needed to prevent privacy violations of various
kinds.
The results in the literature vary according to several parameters,
including the size of the database, the size of the universe from
which data elements are drawn, the ``amount'' of privacy desired, and
for the purposes of the current work, the arity of the query.
In this paper we sharpen and unify these bounds.  
Our foremost result combines the techniques of Hardt and Talwar~\cite{HT10}
and McGregor {\it et al.}~\cite{MMPRTV10} to obtain linear lower
bounds on distortion when providing differential privacy for 
a (contrived) class of low-sensitivity queries. (A query has low sensitivity
if the data of a single individual has small effect on the answer.)
Several structural results follow as immediate corollaries:
\begin{itemize}
\item
We separate so-called {\em counting} queries from arbitrary {\em
  low-sensitivity} queries, proving the latter requires more noise, or
distortion, than does the former;
\item
We separate $(\eps,0)$-differential privacy from 
its well-studied relaxation $(\eps,\delta)$-differential privacy, 
even when
$\delta \in 2^{-o(n)}$ is negligible in the size $n$ of the database,
proving the latter requires less distortion than the former;
\item
We demonstrate that $(\eps,\delta)$-differential privacy is much weaker
than $(\eps,0)$-differential privacy  in terms of mutual information of the transcript of the mechanism with the database,
even when
$\delta \in 2^{-o(n)}$ is negligible in the size $n$ of the database.
\end{itemize}
We also simplify the lower bounds on noise for counting queries in \cite{HT10} and also make them unconditional. Further, we use a characterization of $(\epsilon,\delta)$ differential privacy from \cite{MMPRTV10} to obtain lower bounds on the distortion needed to ensure
$(\eps,\delta)$-differential privacy for $\epsilon,\delta > 0$.

 Next, we revisit the LP decoding argument of \cite{DMT07} and combine it with recent results of Rudelson \cite{Rud11} to show that for some specific $\eta>0$, if the $\ell$-way marginals are released such that at least $1-\epsilon$ fraction of the entries have $o(\sqrt{n})$ noise, then a very minimal notion of privacy called attribute privacy is violated.  This improves on a recent result of Kasiviswanathan \etal \cite{KRSU10} where the same conclusion was shown assuming that all the entries have $o(\sqrt{n})$ noise. 
 
 Finally, we extend the original lower bound of \cite{DN03}  to prevent blatant non-privacy to the case when the universe size is smaller than the size of the database. As we show, the lower bound on the noise required to prevent blatant non-privacy becomes larger as the size of the universe decreases. \end{abstract}
 
\thispagestyle{empty}
\vfill

\end{titlepage}
\newpage
\section{Introduction}
This is a paper about private data analysis, in which a trusted
curator holding a confidential database responds to real vector-valued
queries.  Specifically, we focus
on the practice of ensuring privacy for the database
elements by adding appropriately generated random noise to the
answers, releasing only these {\em noisy} responses.  A line of study
initiated by Dinur and Nissim examines the amount of distortion needed
to prevent privacy violations of various kinds~\cite{DN03}.  Dinur and
Nissim did not have a definition of privacy; rather, they had a notion
that has come to be called {\em blatant non-privacy}; the modest goal,
then, was to add enough distortion to avert blatant non-privacy.
Since that time, the community has raised the bar by definining (and
achieving) powerful and comprehensive notions of
privacy~\cite{DN03,DMNS06,DKMMN06}, and the goal has been to
preserve $(\eps,0)$-differential privacy and
its relaxation, $(\eps,\delta)$-differential privacy.
A final goal considered herein, {\em attribute privacy}, 
has a more complicated description, but may be thought of as
preventing blatant non-privacy for a single data attribute~\cite{KRSU10}
in the presence of a certain kind of contingency table query.

The results in the literature vary according to several parameters,
including the number $n$ of elements in the database, 
the size $d$ of the universe from
which data elements are drawn, the ``amount'' and type
of privacy desired, and
for the purposes of the current work, the arity $k$ of the query.
In this paper we strengthen and unify these bounds.  

As corollaries of our work, we obtain several ``structural'' results
regarding different types of privacy guarantees:
\begin{itemize}
\item
We separate so-called {\em counting} queries from arbitrary {\em
  low-sensitivity} queries, proving the latter requires more noise, or
distortion, than does the former;
\item
We separate $(\eps,0)$-differential privacy from 
its well-studied relaxation $(\eps,\delta)$-differential privacy, 
even when
$\delta \in 2^{-o(n)}$ is negligible in the size $n$ of the database,
proving the latter requires less distortion than the former;
\item
We demonstrate that $(\eps,\delta)$-differential privacy is much weaker
than $(\eps,0)$-differential privacy  in terms of mutual information of the transcript of the mechanism with the database
even when
$\delta \in 2^{-o(n)}$ is negligible in the size $n$ of the database.
\end{itemize}
We also simplify the lower bounds on noise for counting queries in \cite{HT10} and also make them unconditional removing a technical assumption on the mechanism present in their paper.  Next, we use a characterization of $(\epsilon,\delta)$ differential privacy from \cite{MMPRTV10} to obtain lower bounds on the distortion needed to ensure
$(\eps,\delta)$-differential privacy for $\epsilon,\delta > 0$. We remark that \cite{KRSU10} also obtain  quantitatively similar lower bounds on the distortion required to maintain $(\epsilon,\delta)$ differential privacy for the class of $\ell$-way marginals though their proof technique is very different and arguably much more complicated. 

After this, we use results of Rudelson \cite{Rud11} and combine it with LP decoding to show that  attribute privacy is violated if $\ell$-way marginals are released with at least $1-\eta$ fraction of these marginals are released with $o(\sqrt{n})$ noise for some $\eta>0$. The results and the technique  in \cite{KRSU10} required $\eta=0$ making our results more powerful. Finally, we extend the results of \cite{DN03} to the case of small universe size achieving stronger lower bounds to prevent blatant non-privacy.

To describe our results even at a high level we must outline
the privacy-preserving database model, the notion of {\em distortion}
or {\em noise} that may be employed in order to preserve privacy,
and the meaning of the goals of the adversary: blatant non-privacy,
violation of $(\eps,0)$-differential privacy, violation of $(\eps,\delta)$-
differential privacy, and attribute non-privacy.

Typically, the curator of a database receives questions to which it responds with potentially noisy answers.
There are two possible settings here. One is that the queries are received by the curator one at a time.
The other situation is that all the queries are received by the curator at once and it then publishes (noisy) answers to all 
of them at once. The former is called the interactive setting and the latter is called the non-interactive setting. All our lower bounds are in
the non-interactive setting making them applicable to the interactive setting as well. 

We now formally describe a database and a query : A database $X$ is an element of $(\mathbb{Z}^+)^d$ . Here $d$ is called the universe size and intuitively refers to the number of types of elements present in the database. Also, for a database $X$, $n=\sum_{i=1}^d X_i $ is defined as the size of the database and refers to the number of elements in the database. Note that we are representing databases as histograms. 
A query (of arity $k$) is a map $F : (\mathbb{Z}^+)^d \rightarrow \mathbb{R}^k$ such that $\forall i \in [k]$, $\forall x,y \in (\mathbb{Z}^+)^d$, 
$|F(x+y)_i-F(x)_i| \le 1$ if $\Vert y \Vert_1 =1$. In other words, every coordinate of the map $F$ is $1$-Lipschitz. 
We say $F$ is a counting query if $F$ is a linear map. 
The meaning of $d,k,n$ throughout the paper shall be the same as above unless mentioned otherwise. 

We now formally introduce the definition of mechanism and privacy. 
\begin{definition}
Let $\mathcal{F}$ be a family of queries such that $\forall F \in \mcal{F}$, $F:(\mathbb{Z^+})^d \rightarrow \mathbb{R}^k$.  Then, a mechanism $M:(\mathbb{Z^+})^d \times \mcal{F} \rightarrow \mu(\mathbb{R}^k)$ where  $\mu(\mathbb{R}^k)$ 
is simply the set of probability distributions over $\mathbb{R}^k$. On being given a query $F \in \mcal{F}$ and a database $x \in (\mathbb{Z^+})^d$, the curator samples  $z$ from  the probability distribution $M(x,F)$ and returns $z$. 
\end{definition} 
We next state the definition of $\epsilon$-differential privacy (introduced by Dwork \etal in \cite{DMNS06}) and $(\epsilon,\delta)$-differential privacy (introduced by Dwork \etal in \cite{DKMMN06}). 
\begin{definition}
For a family of queries $\mcal{F}$, a mechanism $M:(\mathbb{Z^+})^d \times \mcal{F} \rightarrow \mu(\mathbb{R}^k)$ is said to be $\epsilon$-differentially private if for every $x, y \in (\mathbb{Z^+})^d$ such that $\Vert x-y\Vert_1 \le 1$, every measurable set $S \subseteq \mathbb{R}^k$ and $\forall F \in \mcal{F}$, the following holds : Let $M(x,F) = M_{x,F}$ and $M(y,F) = M_{y,F}$  and for a probability distribution $\Gamma$, let $\Gamma(S)$ denote the probability of set $S$ under $\Gamma$. Then, 
$$
 2^{-\epsilon} \le \frac{M_{x,F}(S) }{ M_{y,F}(S)} \le 2^{\epsilon}
$$
The mechanism is said to be $(\epsilon,\delta)$-differentially private if 
$$
 2^{-\epsilon}\cdot M_{y,F}(S) - \delta \le M_{x,F}(S)   \le 2^{\epsilon} \cdot M_{y,F}(S) +\delta
$$
Typically, $\delta$ is set to be negligible in $n,k$.  
\end{definition}
We remark that we do not define the notion of noise very precisely here as the notion of noise depends on the context. However, in the context of differential privacy, we use the following definition of noise. 
\begin{definition}
For a family of queries $\mcal{F}$, a mechanism $M:(\mathbb{Z^+})^d \times \mcal{F} \rightarrow \mu(\mathbb{R}^k)$ is said to add noise (at  most) $\eta$ if with high probability (say $0.99$) over the randomness of $M$, $\Vert M(x,F) - F(x) \Vert_{\infty} \le \eta$. 
\end{definition}
While differential privacy is a very strong notion of privacy, sometimes one can show that even very modest definitions of privacy get violated. One such notion is that of blatant non-privacy. We say that 
a mechanism $M$ for answering $F$ over databases of size $n$ and universe size $d$ is blatantly non-private, if there is an attack $A$ such that w.h.p. over the answer $y$ returned by the mechanism $M$, $A(y)$ differs from the database only at $o(1)$ fraction of the places. Yet another very weak notion of privacy that is interesting to us is that of attribute non-privacy. The formal definition follows : 
\begin{definition}
For a query $F \in \mathcal{F}$, a mechanism $M : (\B^d)^n \times \mcal{F} \rightarrow \mathbb{R}^k$ is said to be attribute non-private if there exists $Y \in (\B^{d-1})^n$ and an algorithm $A$ such that for every $x \in \B^n$, 
$$
\Pr_{z \in M(Y\circ x,F)} [A(z) = x' : \Vert x-x'\Vert _1 = o(\Vert x\Vert_1)] \ge 1/10
$$
where $Y \circ x$ simply denotes the obvious concatenation of $Y$ and $x$. $A$ need not be computationally efficient and the constant $1/10$ is arbitrary and can be replaced by any positive constant.  \end{definition} 

We show the following results : 
\begin{enumerate}
\item Combining techniques from \cite{HT10} and \cite{MMPRTV10}, we obtain tight lower bounds on the noise  for arbitrary (non-counting) low-sensitivity queries for any $(\eps,0)$-differentially
private mechanism. Given positive results of Blum, Ligett, and Roth~\cite{BLR08}, this separates non-counting queries from counting queries, proving that
the former require more distortion than the latter for maintaining
differential privacy. Also, given the positive results of \cite{DKMMN06} for arbitrary low-sensitivity queries,
this separates $(\eps,\delta)$-differential privacy from $(\eps,0)$-differential
privacy, where $\delta = \delta(n,k)$ denotes a function negligible in its argument. We also use this technique to show that the guarantees in terms of information content is drastically weaker for an $(\epsilon,\delta)$ differentially private protocol as compared to an $\epsilon$-differentially private protocol. Our technique also simplifies the \emph{volume-based} lower bounds on noise for counting queries in \cite{HT10}. In addition, we also make the lower bounds unconditional. The lower bound in \cite{HT10} required the mechanism to be defined on ``fractional'' databases \ie on $(\mathbb{R}^+)^d$ as opposed to just $(\mathbb{Z}^+)^d$ while we do not have any such restrictions. 
\item
We give tight lower bounds on noise for ensuring
$(\eps,\delta)$-differential privacy for $\delta > 0$.  
This proof relies on a lemma due to~\cite{MMPRTV10} showing that
$(\eps,\delta)$-differentially private mechanisms yield a certain kind
of unpredictable source.  On the other hand, any mechanism that is
blatantly non-private cannot yield an unpredictable source.
Thus, if the noise is insufficient to prevent blatant non-privacy then
it cannot provide $(\eps,\delta)$-differential privacy. We subsequently use the lower bounds of \cite{DN03, DMT07} for preventing blatant non-privacy to get lower bounds on the distortion for $(\epsilon,\delta)$ differential privacy. 
\item
We revisit the LP decoding attack of Dwork, McSherry, and
Talwar~\cite{DMT07}, observing that any linear query matrix yielding a
Euclidean section suffices for the attack.  
The LP decoding attack succeeds
even if a certain constant fraction of the responses have wild noise.
Armed with the connection to Euclidean sections, and a recent 
result of Rudelson~\cite{Rud11} bounding from below the least singular
value of the Hadamard product of certain i.i.d.\, matrices,
we qualitatively strengthen a lower bound
of Kasiviswanathan, Rudelson, Smith, and Ullman~\cite{KRSU10} on the noise
needed to avert attribute non-privacy in $\ell$-way marginals release
by making the attack resilient to a constant fraction of wild responses.
\item
We show new lower bounds for blatant non-privacy for the case when the size of the universe is smaller than the size of the database. In particular, we show that there is a counting query such that if the distortion added on at least $1/2+\eta$ fraction of the answers is bounded by $o(n/\sqrt{d})$ (for some $\eta>0$), then there is an attack which recovers a database different from the original database by $o(n)$. Our analysis makes use of a result on large deviation of Rademacher sums. 
\end{enumerate}
\section{Lower bound by volume arguments}\label{sec:diffpriv}
We now recall the volume based argument of Hardt and Talwar \cite{HT10} to show lower bounds on the noise required for $\epsilon$ differential privacy. 
\begin{theorem}\label{thm:thm1}
Assume $x_1, \ldots, x_{2^s} \in \mathbb{(Z^{+})}^d$ such that $\forall i$, $\Vert x_i \Vert_1 \leq n$ and for $i \not =j$, $\Vert x_i - x_j\Vert_1 \le \Delta$. Further, let $F : {(\mathbb{Z}^{+})}^d \rightarrow \mathbb{R}^k$ such that for any $i \not = j$, $ \Vert F(x_i)  - F(x_j) \Vert_{\infty} \geq \eta$. If $\Delta \leq (s-1)/\epsilon$, then any mechanism which is $\epsilon$-differentially private for the query $F$ on databases of size $n$  must add  noise $\eta/2$. 
\end{theorem}
While the line of reasoning in the proof is same as that of \cite{HT10}, we do the proof here as  the argument in \cite{HT10} works only for counting queries \ie when $F$ is a linear transformation. On the other hand, the statement and proof of our result works for any query $F$. 
\begin{proof}
Consider the $\ell_{\infty}$ balls of radius $\eta/2$ around each of the $F(x_i)$. By the hypothesis, these balls are disjoint. Now assume, any mechanism $M$ which adds noise $\eta/2$  and consider any $x_i$. Then, because all the balls are disjoint, we have that there is some $j \not = i$ such that if $S$ is the $\ell_{\infty}$ ball of radius $\eta/2$ around $F(x_j)$, then 
$$
\Pr_{z \in M(x_i,F)} [ z \in  S ] \le 2^{-s}
$$
However, we can also say  that because the noise added by the mechanism $M$ is at most $\eta$, 
$$
\Pr_{z \in M(x_j,F)} [ z \in  S ] \ge 1/2
$$
Also, because the mechanism $M$ is $\epsilon$-differentially private and $\Vert x_i - x_j \Vert_1 \le \Delta$, then 
$$
\frac{\Pr_{z \in M(x_i,F)} [ z \in  S ] }{\Pr_{z \in M(x_j,F)} [ z \in  S ] } \geq 2^{- \epsilon \cdot \Delta}
$$
This leads to a contradiction if $\Delta \le (s-1)/\epsilon$ thus proving the assertion. 
\end{proof}
\subsection{Linear lower bound for arbitrary queries}
In this subsection, we prove the following theorem. 
\begin{theorem}\label{thm:main}
For any $k , d, n  \in \mathbb{N}$ and $1/40\geq \epsilon>0$, where $ n \geq \min\{k/\epsilon, d/\epsilon\}$, there is a query $F: (\mathbb{Z}^+)^d \rightarrow \mathbb{R}^k$ such that any mechanism $M$ which is $\epsilon$-differentially private adds noise $\Omega(\min \{ d/\epsilon, k/\epsilon \})$.

If $\epsilon>1$,  then there is a query $F: (\mathbb{Z}^+)^d \rightarrow \mathbb{R}^k$ such that any mechanism $M$ which is $\epsilon$-differentially private adds noise $\Omega(\min \{ d/(\epsilon \cdot 2^{5\epsilon}), k/\epsilon \})$ as long as $n \geq \min\{k/\epsilon, d/(\epsilon \cdot 2^{5\epsilon})\}$
\end{theorem}
Before starting the proof, we make a couple of observations. First of all, note that the statement of the theorem does not give any lower bound for $1\ge \epsilon > 1/40$. However, any mechanism which is $\epsilon$-differentially private for $\epsilon$ in the aforementioned range is also $\epsilon'$-differentially private for $\epsilon'=10/9$. Hence, the noise lower bounds for $\epsilon'$-differential privacy for $\epsilon'=10/9$ are also applicable for the range of $1\ge \epsilon > 1/40$. It is easy to see that up to constant factors, the lower bounds with $\epsilon'=10/9$ are optimal for $\epsilon$ in the aforementioned range. 

Secondly, we note that it is enough to add noise $O(k/\epsilon)$ to maintain $\epsilon$-differential privacy (using the Laplacian mechanism). Also, because the databases are of size $n$, it is enough to add noise $O(n)$ to maintain $\epsilon$-differential privacy for any $\epsilon \ge 0$. Thus, as long as $k = O(d)$, our lower bounds are tight up to constant factors. Next, we do the proof of Theorem~\ref{thm:main}.

\begin{proof}
Our proof strategy is to construct a set of databases and a query which meets the conditions stated in the hypothesis of Theorem~\ref{thm:thm1} and then get the desired lower bound on the noise. We first deal with the case when $0 <\epsilon<1/40$.  Let $\ell = \min\{d,k \}$.  We can now use Claim~\ref{clm:clm1}  to construct $2^{s}$ databases $x_1, \ldots , x_{2^s}$  (for $s=\ell/400$) such that $x_i  \in (\mathbb{Z}^{+})^d$ with the property that  $\forall i \not = j$, $\Vert x_i  - x_j \Vert_1 \ge n'/10 $ and
$\Vert x_i\Vert_1 \leq n'$ where $n'=\ell/(1280\epsilon)$ (Application of Claim~\ref{clm:clm1} uses $d'=\ell/320$). Note that our databases are of size bounded by $n' \le n$. We now describe a mapping $\mcal{L}:   (\mathbb{Z}^{+})^d \rightarrow \mathbb{R}^{2^s}$ which is related to a construction in \cite{MMPRTV10}. The mapping is as follows : 
\begin{itemize}
\item For every $x_i$, there is a coordinate $i$ in the mapping. 
\item The $i^{th}$ coordinate of $\mcal{L}(z)$ is $\max \{ n'/30- \Vert x_i -z \Vert_1, 0 \}$. 
\end{itemize}
\begin{claim}
The map $\mcal{L}$ is $1$-Lipschitz  \ie  if $\Vert z_1 - z_2 \Vert _1 =1$, then $\Vert \mcal{L}(z_1) - \mcal{L} (z_2) \Vert_1 \le 1$. 
\end{claim}
\begin{proof}
We observe that for any $z_1, z_2$ such that $\Vert z_1 - z_2\Vert  \le 1$, if $A$ denotes the set of coordinates where at least one of $\mcal{L}(z_1)$ or $\mcal{L}(z_2)$ are non-zero, then $A$ is either empty or is a singleton set. Given this, the statement in the claim is obvious, since the mapping corresponding to any particular coordinate is clearly $1$-Lipschitz. 
\end{proof}
We now describe the queries. Corresponding to any $r \in \{ -1,1 \}^{2^s}$,  we define $f_r: (\mathbb{Z}^+)^d \rightarrow \mathbb{R}$, as 
$$
f_r(x) = \sum_{i=1}^d \mathcal{L}(x)_i \cdot r_ i 
$$
Now, we define a random map $F : (\mathbb{Z}^+)^d \rightarrow \mathbb{R}^k$ as follows. Pick $r_1, \ldots, r_k \in \{ -1,1 \}^{2^s}$ independently  and uniformly at random  and define $F$ as follows  : 
$$
F(x) = (f_{r_1}(x), \ldots, f_{r_k}(x))
$$
Now consider any $x_h,x_j \in S$ such that $h \not =j$. Because of  the way $\mathcal{L}$ is defined, it is clear that for any $r_i$, $$\Pr_{r_i} [|f_{r_i}(x_h) - f_{r_i}(x_j)| \geq n'/15] \geq 1/2$$
A basic application of the Chernoff bound implies that 
$$\Pr_{r_1,\ldots,r_k} [\textrm{For at least }1/10 \textrm{ of the }r_i\textrm{'s},\quad|f_{r_i}(x_h) - f_{r_i}(x_j)| \geq n'/15] \geq 1 - 2^{-k/30}$$
Now, note that the total number of pairs $(x_i,x_j)$ of databases such that $x_i, x_j \in S$ is at most $2^{2s} \leq 2^{\ell/200} \leq 2^{k/200}$. This implies (via a union bound) 
$$
\Pr_{r_1,\ldots,r_k} [ \forall h \not = j, \quad \textrm{For at least }1/10 \textrm{ of the }r_i\textrm{'s},\quad|f_{r_i}(x_h) - f_{r_i}(x_j)| \geq n'/15 ] \geq 1 -2^{-k/40}
$$
This implies that we can fix $r_1,\ldots, r_k$ such that the following is true. 
$$
\forall h \not = j, \quad \textrm{For at least }1/10 \textrm{ of the }r_i\textrm{'s},\quad|f_{r_i}(x_h) - f_{r_i}(x_j)| \geq n'/15$$
This implies that for any $x_h \not = x_j \in S$, $\Vert F(x_h) - F(x_j)\Vert_{\infty} \ge   n'/15$.  In fact, $\Vert F(x_h) - F(x_j)\Vert_{2} \ge   n'\sqrt{k}/150$ which is a much stronger assumption than what we require and is quantitatively similar to the results in \cite{HT10} where they consider $\ell_2$ noise as opposed to $\ell_{\infty}$ noise. 

We can now apply Theorem~\ref{thm:thm1} by putting $\Delta  =2n'$ and $s = \ell/400   > 3 \epsilon  n'$ and $\eta=n'/15$ and observe that $ \Delta \leq (s-1)/\epsilon$ thus proving the result.

We next deal with the case when $\epsilon>1$. This part of the proof differs from the case when $\epsilon<1$ only in the construction of $x_1, \ldots, x_{2^s}$. We also emphasize that had we not insisted on integral databases, our proof would have been identical to the first part. We construct the databases $x_1, \ldots , x_{2^s}$  using combinatorial designs. More precisely,  for some sufficiently large constant $C$,  let $\ell = \min\{d/(C \cdot 2^{5\epsilon}),k \}$.  We can now use Claim~\ref{clm:clm6}  to construct $2^{s}$ databases $x_1, \ldots , x_{2^s}$  (for $s=\ell/400$) such that $x_i  \in (\mathbb{Z}^{+})^d$ with the property that  $\forall i \not = j$, $\Vert x_i  - x_j \Vert_1 \ge n'/10 $ and
$\Vert x_i\Vert_1 \leq n'$ where $n'=\ell/(1280\epsilon)$ (using $d'=\ell/320$ in Claim~\ref{clm:clm6}). Again, we note here that the databases constructed are of size $n'$. 

From this point onwards, we define the map $\mathcal{L}$ and the query $F$ as we did in the proof of Theorem~\ref{thm:main} and the proof proceeds identically. In particular, we get a query $F:( \mathbb{Z}^+)^d \rightarrow \mathbb{R}^k$ such that for any $i \not =j$,  $\Vert F(x_i) - F(x_j) \Vert_2 \geq n' \sqrt{k}/150$. As before, we can now apply Theorem~\ref{thm:thm1} by putting $\Delta  =2n'$ and $s = \ell/100   > 3 \epsilon  n'$ and $\eta=n'/15$ and observe that $ \Delta \leq (s-1)/\epsilon$ thus proving the result 
\end{proof}

For the subsequent part of this paper, we only consider lower bounds on $\epsilon$-differential privacy for $0<\epsilon<1$ as opposed to $\epsilon>1$. This is because the privacy guarantees one gets becomes unmeaningful when $\epsilon$ is large. However, we do remark that the results can be carried in a straightforward way to the regime of $\epsilon>1$ using combinatorial designs (like we did for Theorem~\ref{thm:main}). 

\subsubsection*{Consequences of the linear lower bound}
We briefly describe the two consequences of the linear lower bound on the noise  proven in Theorem~\ref{thm:main}. The first is separation of counting queries from non-counting queries. While our separation gives quantitatively the same results as long as $d = k^{O(1)}$ and $n = \Theta(k/\epsilon)$, for simplicity, we consider the setting when $k = d $ and $n = k/\epsilon$.  In this case, Theorem~\ref{thm:main} shows existence of  a (non-counting) query such that maintaining $\epsilon$-differential privacy requires noise $\Omega(n)$. On the other hand, \cite{BLR08} had proven that for any counting query with the same setting of parameters, there is a mechanism which adds noise $\tilde{O}(n^{2/3})$ and maintains $\epsilon$-differential privacy. This shows that maintaining $\epsilon$-differential privacy inherently  requires more distortion in case of non-counting queries than counting queries. 

The next consequence is a separation of $(\epsilon,\delta)$ differential privacy from $(\epsilon,0)$ differential  privacy for $\delta =2^{-o(n)}$. We note that Hardt and Talwar \cite{HT10} had shown such a separation but that was only when $k =O(\log n)$ and $\delta = n^{-O(1)}$. Again, we use the setting of parameters when $k=d$ and $n=k/\epsilon$. The gaussian mechanism of \cite{DKMMN06} shows that to maintain $(\epsilon,\delta)$ differential privacy for any $k$ queries, it sufficies to add noise $O(\sqrt{ k \log(1/\delta) } /\epsilon) =o(n)$. However, Theorem~\ref{thm:main} shows that there is a query which requires adding noise $\Omega(n)$ to maintain $(\epsilon,0)$ differential privacy. 

The last consequence of our result is more indirect and is explained next. 
\subsection{Information loss in differentially private protocols}\label{sec:infoloss}
In \cite{MMPRTV10}, a connection was established between differentially private protocols and the notion of mutual information from information theory. 
In fact, as \cite{MMPRTV10} was dealing with 2-party protocols,  the connection was actually between differentially private protocols and that of information content \cite{BYJKS04,BBCR10} which is a symmetric variant of mutual information  useful in 2-party protocols. In that paper, it was shown that the information content (which simplifies to mutual information in our setting) between transcript of a $\epsilon$-differentially private mechanism and the database vector is bounded by $O(\epsilon n)$. Using the construction used in the previous subsection, we show that in case of $(\epsilon,\delta)$ differentially private protocols (for any $\delta= 2^{-o(n)}$), there is no non-trivial bound on the mutual information between the transcript of the mechanism and the database vector. Thus as far as information theoretic guarantees go, the situation is drastically different for pure differentially private protocols vis-a-vis approximately differentially private protocols. The contents of this subsection are a result of personal communication between the author and Salil  Vadhan \cite{DVa10}.

We first define the notion of mutual information (can be found in standard information theory textbooks).
\begin{definition}
Given two random variables $X$ and $Y$, their mutual information $I(X;Y)$ is defined as 
$$
I(X;Y) = H(X) + H(Y) - H(X,Y) = H(X) - H(X|Y) 
$$
where $H(X)$ denotes the Shannon entropy of $X$. 
\end{definition}
The next claim establishes an upper bound on the mutual information between transcript of a differentially private protocol and the database vector. 
\begin{claim}
Let $F:(\mathbb{Z}^+)^d \rightarrow \mathbb{R}^k$ be a query and  $M : (\mathbb{Z}^+)^d \rightarrow \mu(\mathbb{R}^k)$ be an $\epsilon$-differentially private protocol for answering $F$ for databases of size $n$. If $X$ is a distribution over the inputs in $(\mathbb{Z}^+)^d$, then $I(M(X);X) \le 3\epsilon n$. 
\end{claim}
\begin{proof}
We first note that since the databases are of size bounded by $n$, hence instead of assuming that $\mu$ is a distribution over the inputs $X \in (\mathbb{Z}^+)^d$, we can assume that $\mu$ is a distribution over the inputs $X \in [n]^d$
where $[n] =\{0,1,\ldots, n\}$.  Now, we can apply Proposition~$7$  from \cite{MMPRTV10}.  We note that the aforesaid proposition is in terms of information content for $2$-party protocols but we observe that we can simply make the second party's input as a constant and get that $I(M(X);X) \le 3\epsilon n$.
\end{proof}
Next, we state the following claim which says that for $(\epsilon,\delta)$ differentially private protocols, even for an exponentially small $\delta$, the mutual information between the transcript and the input can be as large as $n(1-\eta)$ for any value of $0<\epsilon,\eta<1$. In other words, an $(\epsilon,\delta)$ differentially private protocol does not imply any effective bound on the mutual information between the input and the transcript even as $\epsilon \rightarrow 0$ and $\delta$ is exponentially small. 

\begin{lemma}\label{lem:mutual}
For $n \in \mathbb{N}$ and $0<\epsilon,\eta<1$, there is a constant $C = C(\epsilon,\eta)>0$ and a distribution $X$ over $(\mathbb{Z}^+)^n$ with a support over databases of size $n$ and   a query $F: (\mathbb{Z}^+)^n \rightarrow \mathbb{R}^k$  and an $(\epsilon,\delta)$-differentially private protocol $M$ for answering $F$ such that $I(X; M(X))  \ge n(1-2\eta)$ if $\delta \ge 2^{-C(\epsilon,\eta) n}$. 
\end{lemma}
\begin{proof}
We first construct $2^s$ vectors in $\{0,1\}^n$ (for $s={n(1-\eta)}$) with the property that for any $x_i, x_j$ $(i \ne j)$, $\Vert x_i - x_j \Vert_1 \ge  \eta^2n/8$.  It is easy to guarantee the existence of such a set of vectors by a simple application of the probabilistic method. The distribution $X$ is simply the uniform distribution over the set $\{x_1, \ldots, x_{2^s}\}$. By construction, all the databases in $X$ are of size bounded by $n$. 

Next, we define the query $F : (\mathbb{Z}^+)^n \rightarrow \mathbb{R}^k$ be defined in the same way as the query $F$ in the proof of Theorem~\ref{thm:main}.  Following, exactly the same calculations, we can show that if we set $k =80 n$, we get a query $F : (\mathbb{Z}^+)^n \rightarrow \mathbb{R}^k$ such that for any $i \not = j$,  $\Vert F(x_i) - F(x_j) \Vert_2 \ge \eta^2 n\sqrt{k}/50$.  
We now recall the Gaussian mechanism of \cite{DKMMN06} which maintains $(\epsilon,\delta)$ differential privacy. 
\begin{lemma}\cite{DKMMN06}  
Let $F: (\mathbb{Z}^+)^d \rightarrow \mathbb{R}^k$ be a query.  Let $Y  = (Y_1, \ldots, Y_k)$ be a distribution over $\mathbb{R}^k$ such that each $Y_i$ is an i.i.d. $\mcal{N}(0,\sigma)$ random variable. Here $\sigma^2 = \frac{k \log(1/\delta)}{\epsilon^2}$.  Then the mechanism $M$ which for a database $x$ and query $F$, which samples $Y_0$ from $Y$ and responds by $F(x) + Y_0$ is an $(\epsilon,\delta)$ differentially private mechanism.
\end{lemma}
Note that for the above mechanism $M$, and database $x$,  if $Z$ is sampled from $M(x)$, then the distribution of $M(x) - F(x)$ is same as $(Y_1, \ldots, Y_k)$ where each $Y_i$ is an i.i.d. $\mcal{N}(0,\sigma)$ random variable. Thus, 
$$
\Vert M(x) - F(x) \Vert_2^2 \sim Y_1^2 + \ldots + Y_k^2
$$
As the following fact shows, the distribution on the right hand side is concentrated around its mean. The fact is possibly well-known but we could not find a reference and hence we prove it in Appendix~\ref{sec:chi}.
\begin{fact}
If $Y_1, \ldots, Y_k$ are i.i.d. $\mcal{N}(0,\sigma)$ random variables, then, 
$$
\Pr_{Y_1, \ldots, Y_k} [ Y_1^2 + \ldots + Y_k^2 > 2 (1+\xi)\cdot k   \cdot \sigma^2 ] \le 2^{ - \frac{k\xi}{2}}
$$
\end{fact}
Using the above fact, we get 
$$
\Pr \left[ \Vert M(x) - F(x) \Vert_2^2  > \frac{2(1+\xi )  k^2 \log(1/\delta)}{\epsilon^2} \right] 
\leq 2^{\frac{-\xi k}{2}}$$
Here the probability is over the randomness of the mechanism.  Putting $\xi=1$ and $\delta = 2^{-C(\epsilon,\eta) n}$ for an appropriate constant $C(\epsilon,\eta)$, we get that 
$$
\Pr \left[ \Vert M(x) - F(x) \Vert_2  > \frac{\eta^2 n \sqrt{k}}{200} \right] 
\leq 2^{-40n}$$
As we know, for any $i \not = j$, $\Vert F(x_i) - F(x_j) \Vert_2 \ge \eta^2 n\sqrt{k}/50$. Hence, with probability at least $1-2^{-n}$ over the randomness of the mechanism,  for any database $x_i \in supp(X)$, if $y$ is sampled from $M(x_i)$,$$
\forall j \not = i \ \ \Vert F(x_j) - y \Vert_2 > \Vert F(x_i) - y \Vert_2
$$
Thus, for any $x_i$, given $M(x_i)$, we can recover $x_i$ with high probability and hence, we can say 
$$
\Pr_{y \sim M(X)}  [ H(X | M(X) =y) =0] >1-2^{-n}
$$
This means that 
$$
H(X |M(X)) \le 2^{-n} n  <1
$$
Recall that $I(X;M(X)) = H(X) - H(X|M(X)) \ge H(X) - 1 = (1-\eta) n -1 \ge (1-2 \eta) n$. This completes the proof of the Lemma~\ref{lem:mutual}.
\end{proof}
\section{Lower bound on noise for counting queries}
In the last section, we proved that to preserve $\epsilon$ differential privacy for $k$ queries, one may need to add $\Omega(k/\epsilon)$ noise  provided $d, n \gg k$. However, these queries were not counting queries. It is interesting to derive lower bounds on noise required  to preserve privacy for counting queries as these are the queries mostly used in practice. While one might initially hope to prove a similar lower bound for counting queries, \cite{BLR08} states 
that there is a  $\epsilon$-differentially private mechanism which adds $\tilde{O}(n^{2/3}/\epsilon)$ noise per query and can answer $O(n)$ counting queries (when $d  =n^{O(1)}$).

Still, Hardt and Talwar \cite{HT10} showed that to answer $k$ counting queries, any mechanism which is $\epsilon$-differentially private must add $\Omega(\min \{k/\epsilon, \sqrt{k \log(d/k)}/\epsilon \})$ noise  (in fact, this is true for $k$ random queries). However, \cite{HT10} make a technical assumption that the  mechanism has a smooth extension which works for ``fractional" databases as well. In other words, they require the domain of the mechanism to be $(\mathbb{R}^+)^d$ as opposed to $(\mathbb{Z}^+)^d$. However, it is not clear if this is always true \ie if given a mechanism which is defined only over true (integral) databases, one can get a mechanism which is defined over ``fractional" databases with similar privacy guarantees. 

Next, we prove the same result without making any such technical assumptions. Again, our constructions are dependent on combinatorial designs \cite{EFF85}. First, we prove the following simple but useful claim. 
\begin{claim}\label{clm:clm3}
Let $a \in \mathbb{Z}$ and  assume $x_1, x_2 , \ldots, x_{2^s} \in (\mathbb{Z}^+)^d$ such that $\forall i$, every entry of $x_i$ is either $0$ or $a$. Also, for every $i \not =\ell$,  $\Vert x_i -x_{\ell} \Vert_1 \ge \Delta$. Then, for $k \ge 20 s$,  there is a linear query $F:  (\mathbb{Z}^+)^d \rightarrow \mathbb{R}^{k}$ such that for every $i,\ell \in [2^s]$ and  $i \not =\ell$, the following holds : 
$$
\Pr_{j \in [k]} [|F(x_i)_j - F(x_{\ell})_j| \ge \Delta'/10] \ge 1/40
$$
where $\Delta' = \sqrt{\Delta \cdot a}$. 
\end{claim}
\begin{proof}
Consider any $x_i, x_{\ell}$ such that $i \not = \ell$. Note that, $z$ defined as $z=x_i -x_{\ell}$ is such that all its entries are $0, \pm a$ and also that $z$ has at least $\Delta/a$ or more non-zero entries. If we choose $r \in \{-1,1\}^{d}$ u.a.r., then note that 
$$
Y=\sum_{i=1}^d z_i \cdot r_i = \sum_{z_i = \pm a} z_i \cdot r_i 
$$
Note that the total number of summands is $\ell'  \geq \Delta/a$ and hence the distribution of the random variable $Y$ is same as  choosing $r' \in \{-1,1\}^{d}$ and considering the random variable
$$
Y' =  a \cdot \left( \sum_{i=1}^{\ell'} r'_i \right)
$$
However using Corollary~\ref{cor:cor1}, we get 

\begin{equation}\label{eq:ref1}
\Pr \left[|Y' | \geq  \frac{\sqrt{\Delta \cdot a} }{10} \right] = \Pr \left[  | \sum_{i=1}^{\ell'} r'_i|  \geq  \frac{\sqrt{\Delta/ a} }{10}\right] \geq \frac{9}{10}
\end{equation}
Now, let us choose $r'_1, \ldots, r'_k$ uniformly and independently at random from $\{-1,1\}^{d}$ and consider the linear query $F: (\mathbb{Z}^+)^d \rightarrow \mathbb{R}^k$ defined as
$$
F(x) = \left(\sum_{j=1}^d x_j \cdot r'_{1j}, \ldots, \sum_{j=1}^d x_j \cdot r'_{kj} \right)
$$
Set $\Delta'= \sqrt{\Delta \cdot a}$. Now,  (\ref{eq:ref1}) and an application of Chernoff bound implies that for any  $x_i,x_{\ell}$ ($i \not = \ell$)  
$$
\Pr_{r'_1,\ldots,r'_k} \left[\Pr_{j \in [k]} [|F(x_i)_j - F(x_{\ell})_j| \ge \Delta'/10] \ge 1/40\right] > 1- 2^{-k/10}
$$
We now observe that the total number of pairs $(x_i, x_\ell)$ ($i \not = \ell$)  is at most $2^{2s} \leq 2^{k/10}$. Applying a union bound, we get that there is some choice of $r'_1,\ldots, r'_k$ (and hence a fixed $F$) such that 
$$ \Pr_{j \in [k]} [|F(x_i)_j - F(x_{\ell})_j| \ge \Delta'/10] \ge 1/40 $$
\end{proof}

We now prove a lower bound on the noise required to maintain privacy for random counting queries. As we have said before, Hardt and Talwar \cite{HT10} proved the same result under an additional assumption that the mechanism defined over integral databases can be smoothly extended to fractional databases as well.
\begin{theorem}
For every $k, d \in \mathbb{N}$ and $1>\epsilon>0$, there is a counting query $F: (\mathbb{Z}^+)^d \rightarrow \mathbb{R}^k$ such that any mechanism which maintains $\epsilon$-differential privacy adds noise $\Omega(\min\{k/\epsilon, \sqrt{k \log (d/k)}/\epsilon\})$. The size of the database \ie $n=O(k/\epsilon)$.
\end{theorem}
\begin{proof}
The proof strategy is to come up with databases meeting the hypothesis of Claim~\ref{clm:clm3} and use Claim~\ref{clm:clm3} to get a counting query $F$. 
We then use Theorem~\ref{thm:thm1} to get a lower bound on the distortion required by any private mechanism to answer $F$. 
We consider two cases : $ k\leq \log d$ and $k > \log d$. 

The first case is trivial : Namely, consider databases $x_1, \ldots, x_{2^{k/20}}$ such that each $x_i = \lfloor(k/80\epsilon)\rfloor \cdot e_i$ where $e_i$ is the standard unit vector in the $i^{th}$ direction. This is possible as there are $d \ge 2^k$ different unit vectors.  Note that for any $i \not = \ell$, $\Vert x_i -x_{\ell} \Vert_1 =2 \cdot \lfloor k/(80\epsilon) \rfloor$. We can now apply Claim~\ref{clm:clm3} and get that there is a linear query $F:  (\mathbb{Z}^+)^d \rightarrow \mathbb{R}^k$ (using $\Delta = 2\cdot \lfloor k/(80 \epsilon) \rfloor$ and $a = \lfloor k/(80 \epsilon) \rfloor$) such that 
$$ \Pr_{j \in [k]} \left[|F(x_i)_j- F(x_\ell)_j| \ge \frac{\sqrt{2}}{10} \lfloor k/(80 \epsilon) \rfloor \ \geq \frac{k}{800 \epsilon}\right] \ge 1/40 $$
We see that there are $2^{k/20}=2^s $ databases which differ by exactly $2 \cdot \lfloor k/(80\epsilon) \rfloor = \Delta$. Note that $\Delta \leq (s-1)/\epsilon$.  Hence we can apply Theorem~\ref{thm:thm1} to note that to maintain $\epsilon$-differential privacy, any mechanism  needs to add $k/(800 \epsilon)$ noise. In fact, we note that the $\ell_2$ error of the answer returned by the mechanism  needs to be $\Omega(k^{3/2}/ \epsilon)$ which is quantitatively the same as the result in \cite{HT10}. 

The second case is  slightly more complicated.  We use Claim~\ref{clm:clm0} to construct $x_1,\ldots, x_{2^{k/20}} \in (\mathbb{Z}^+)^d$ with the following properties :
\begin{itemize}
\item Every entry of any of the $x_i$'s is either $0$ or  $a \in \mathbb{Z}$ such that $a \geq (\log(d/k)/160\epsilon)$. 
\item $\forall i$, $\Vert x_i \Vert_1 \leq k/80\epsilon$ and $\forall i \not = j$, $\Vert x_i -x_j \Vert_1 \geq k/160\epsilon$
\end{itemize}
Again, we can apply Claim~\ref{clm:clm3} and get that there is a linear query $F:  (\mathbb{Z}^+)^d \rightarrow \mathbb{R}^k$ (using $\Delta \ge k/(160 \epsilon)$ and $a \geq  (\log(d/k)/160\epsilon)$) such that $\forall i \not= \ell$
$$ \Pr_{j \in [k]} \left[|F(x_i)_j- F(x_\ell)_j| \ge \frac{1}{10} \cdot\frac{ \sqrt{k \log (d/k) } }{{160\epsilon}} \right] \ge 1/40 $$
Again, we have $2^{k/20}$ databases which differ by at most $k/(40 \epsilon)$ and hence we can apply Theorem~\ref{thm:thm1} to get that to maintain $\epsilon$-differential privacy, any mechanism needs to add $\Omega\left(\frac{ \sqrt{k \log (d/k) } }{\epsilon}\right)$ noise.  
\end{proof}
\section{Lower  bounds for approximate differential privacy} 
In this section, we prove lower bounds on the noise required to maintain $(\epsilon,\delta)$ differential privacy for $\epsilon,\delta>0$. Our lower bounds are valid for any positive $\delta>0$ and are in fact tight for a constant $\epsilon$ and $\delta$. We note that a quantitatively similar lower bound was proven for the class of $\ell$-way marginals by \cite{KRSU10} though our proof (for random queries) is arguably much simpler.

  In this section, we consider databases which are elements of $\B^n$ or in other words we consider the case when the universe size $d=n$ and the databases are allowed to have exactly one element of each type. We note that restricting databases to bit vectors is a well-considered model in literature including \cite{DN03,DMT07,MMPRTV10} among others.

We prove the following theorem.
\begin{theorem}\label{thm:epsdel}
For any $n \in \mathbb{N}$, $\epsilon>0$ and $1/20>\delta>0$, there exist positive constants $\alpha,\gamma$ and $\eta$ 
 such that  there is a counting query $F : \B^n \rightarrow \mathbb{R}^k$ with $k=\alpha n$ such that any mechanism $M$ that satisfies 
$$
\Pr_M [ \Pr_{i \in [k]} [|M(x,F)_i - F(x)_i|\le \eta \sqrt{n} ] \geq 1/2+\gamma] \geq 3\sqrt{\delta}
$$
is not $(\epsilon,\delta)$ differentially private.  In other words, any mechanism $M$ which with significant  probability \ie $3 \sqrt{\delta}$ answers at least $1/2+\gamma$ fraction of the $k$ queries with at most $\eta \sqrt{n}$ noise, is not $(\epsilon,\delta)$ differentially private. 
\end{theorem}

An immediate corollary is that there exists a positive constant $\alpha$ and a counting query $F : \B^n \rightarrow \mathbb{R}^k$ where $k = \alpha n$  such that any mechanism which adds $o(\sqrt{n})$ noise is not $(\epsilon,\delta)$ differentially private for $\epsilon>0$ and $\delta<1/20$.

To do the proof of Theorem~\ref{thm:epsdel}, we first need to introduce some definitions previously discussed in \cite{MMPRTV10}. We do note that the paper \cite{MMPRTV10} deals with the two-party setting but the relevant definitions and the lemma we use here easily extend to the standard (curator-client) setting of privacy. \begin{definition}
A random variable $Y \in \B^{n}$ is said to be $\delta$-approximate  strongly $\alpha$-unpredictable bit source (for $\alpha\ge1$) if with probability $1-\delta$ over $i \in [n]$ and $(y_1, \ldots,y_{i-1},y_i, y_{i+1},\ldots,y_n) \leftarrow Y $ 
\begin{displaymath}
 \frac{1}{\alpha} \leq \frac{\Pr[Y_i=1 | Y_1 = y_1 , \ldots, Y_{i-1}=y_{i-1}, Y_{i+1}=y_{i+1},\ldots,Y_n=y_n]}{\Pr[Y_i=0 | Y_1 = y_1 , \ldots, Y_{i-1}=y_{i-1}, Y_{i+1}=y_{i+1},\ldots,Y_n=y_n]}\leq \alpha
\end{displaymath}
\end{definition}
The next lemma (proven in \cite{MMPRTV10} for the two-party setting) roughly says that for any $(\epsilon,\delta)$ private mechanism, conditioned on the transcript of the mechanism, the distribution of the database is a $\delta$-approximate strong $2^{\epsilon}$-unpredictable source. More precisely, we have the following lemma. 
\begin{lemma}\label{lem:lem23}
Let $F : \B^n \rightarrow \mathbb{R}^k$ be a query and $M$ be a $(\epsilon,\delta)$-differentially private mechanism for answering $F$. Let $X$ be the uniform distribution over $\B^n$ and $\Gamma$ be the probability distribution over the transcripts of $M(x)$ when $x$ is drawn from $X$. Then for any $\mu>0$ and $t \leftarrow \Gamma$, the distribution $X|_{\Gamma=t}$ is $\delta_t$ approximate strongly $2^{\epsilon+\mu}$-unpredictable sources such that 
$$\mathop{\mathbb{E}}_{t \in \Gamma}  \ [\delta_t] \le  2\delta \cdot \frac{ 1+ e^{-\epsilon - \mu}}{1-e^{-\mu}}$$.

\end{lemma}
The above lemma trivially follows from Lemma~20 of \cite{MMPRTV10} (full version) and hence we do not prove it here. 
Before, proving Theorem~\ref{thm:epsdel}, we need to recall the following theorem from \cite{DMT07} (Theorem 24 in the paper).
\begin{theorem}\label{thm:thmDMT}
For any $\gamma>0$ and any $\nu=\nu(n)$, there is a  constant $\alpha=\alpha(\gamma)>0$  such that for $k=\alpha n$, there is a counting query $F: \B^n \rightarrow \mathbb{R}^k$ and an algorithm $A$ such that given $\tilde{y}$ which satisfies
$$
\Pr_{i \in [k]} [|\tilde{y}_i - F(x)_i| \leq \nu] \ge \frac{1}{2} + \gamma
$$
the output of $A$ on $\tilde{y}$ \ie $ A(\tilde{y}) =x'$ such that $x' \in \B^n$ and $\Vert x -x' \Vert_1 \leq \frac{4\nu^2}{\gamma^2}$\end{theorem}
The following corollary follows immediately from Theorem~\ref{thm:thmDMT}. 
\begin{corollary}\label{cor:cor47}
For any $\delta'>0$, there are positive constants  $\gamma=\gamma(\delta'), \eta=\eta(\delta'), \alpha=\alpha(\delta')$  such that for $k=\alpha n$, there is a counting query $F: \B^n \rightarrow \mathbb{R}^k$ and an algorithm $A$ such that given $\tilde{y}$ which satisfies
$$
\Pr_{i \in [k]} [|\tilde{y}_i - F(x)_i|  \le \eta \sqrt{n}] \ge \frac{1}{2} + \gamma
$$
the output of $A$ on $\tilde{y}$ \ie $ A(\tilde{y}) =x'$ such that $x' \in \B^n$ and $\Vert x -x' \Vert_1\leq \delta' n$. 
\end{corollary}
We now prove Theorem~\ref{thm:epsdel}.
\begin{proof}[of Theorem~\ref{thm:epsdel}] 

Let $X$ denote the uniform distribution over $\B^n$. First, using Lemma~\ref{lem:lem23}, we get  that over the randomness of the mechanism $M$ and the choice of $x \in X$, if we sample a transcript $t$ from $M(x,F)$, then for any positive $\mu$, the distribution $X|_{M(x,F)= t}$ is  a $\delta_t$-approximate strongly $2^{\epsilon + \mu}$-unpredictable sources where $\delta_t$ satisfies 
$$\mathop{\mathbb{E}}_{t \in M(x,F)}  \ [\delta_t] \le  2\delta \cdot \frac{ 1+ e^{-\epsilon - \mu}}{1-e^{-\mu}}$$.
Clearly, we can put $\mu=10$ and get that the distribution $X|_{M(x,F)= t}$ is  a $\delta_t$-approximate strongly $2^{\epsilon + 10}$-unpredictable sources where $\mathop{\mathbb{E}}_{t \in M(x,F)}  \ [\delta_t] \le 3 \delta$. By an application of Markov's inequality, we get that with probability $1-2\sqrt{\delta}$ over the choice of $x$ and the randomness of the mechanism $M$, the distribution $X|_{M(x,F)= t}$ is $2\sqrt{\delta}$-approximate strongly $2^{\epsilon + 10}$-unpredictable source.

We now apply corollary~\ref{cor:cor47}. In particular, we put $\delta'=\sqrt{\delta}$ and get that for some positive $\gamma,\eta,\alpha$ (which are functions of $\delta'$ and hence $\delta$), there is a counting query $F:\B^n \rightarrow \mathbb{R}^{\alpha n}$ and  an algorithm $A$ such that given $\tilde{y}$ which satisfies
$$
\Pr_{i \in [k]} [|\tilde{y}_i - F(x)_i|  \le \eta \sqrt{n}] \ge \frac{1}{2} + \gamma
$$
the output of $A$ on $\tilde{y}$ \ie $ A(\tilde{y}) =x'$ such that $x' \in \B^n$ and $\Vert x -x' \Vert_1\leq \sqrt{\delta} \cdot n$. Now, consider a mechanism $M$ which satisfies 
$$
\Pr_M [ \Pr_{i \in [k]} [|M(x,F)_i - F(x)_i|\le \eta \sqrt{n} ] \geq 1/2+\gamma] \geq \beta
$$
for $\beta=3\sqrt{\delta}$. Clearly such a mechanism $M$ is not $(\epsilon,\delta)$ differentially private because with probability at least $\beta = 3\sqrt{\delta}$, the algorithm $A$ will be able to predict at least $1- \sqrt{\delta}$ fraction of the positions which contradicts that with probability $1-2\sqrt{\delta}$, the distribution $X|_{M(x,F)= t}$ is a $2\sqrt{\delta}$ -approximate strongly $2^{\epsilon + 10}$-unpredictable source.
\end{proof}

\section{LP decoding, Euclidean sections and hardness of releasing $\ell$-way marginals}
In this section, we consider attacks on privacy  using linear programming. In particular, we use the technique of LP decoding (previously used in \cite{DMT07} in context of privacy) to give attacks which violate even minimal notions of privacy 
when $1- \epsilon_0$ (for some $\epsilon_0>0$) fraction of the queries are released with insufficient noise. We do this by establishing a connection between Euclidean sections and use of LP decoding in context of privacy which does not seem to have explicitly  appeared in the  literature before.  We remark that the relation between LP decoding and Euclidean spaces is very well known in context of compressed sensing \cite{CRTV05}. However, in case of privacy, 
the adversary is allowed to add small error to say $99\%$ of the entries and arbitrary error to the remaining $1\%$ of the entries. In context of compressed sensing however, the adversary is allowed to add error to only $1\%$ of the entries. 

We first describe how to use linear programming in context of privacy. Assume $x \in \mathbb{Z^{+}}^d$ is a database and $A: \mathbb{R}^d \rightarrow \mathbb{R}^{k}$ is a linear map  which represents a counting query with arity $k$ made on the database $x$.  Further, the right set of answers is given by $y = A \cdot x$. (To make sure that the queries are $1$-Lipschitz, all the entries of $A$ come from $[-1,1]$.) Suppose, $\tilde{y} \in \mathbb{R}^k$ is the answer returned by the mechanism. Then, consider the following optimization problem (which can be written as a linear program) :
\begin{equation}\label{eq:eqnLP}
\textrm{Minimize }\Vert y - \tilde{y}\Vert_1 \textrm{ subject to } y=A \cdot \tilde{x}
\end{equation}
We would like to understand the conditions such that the solution to the above linear program, call it $\tilde{x}$, is such that $\Vert x - \tilde{x} \Vert_1 $ is small.  To state our main theorem, we need to first define  a Euclidean section
and list some of its basic properties. 
\begin{definition}
$V \subseteq \mathbb{R}^k$ is said to be a $(\delta,d,k)$ euclidean section if $V$ is a linear subspace of dimension $d$ and for every $x \in V$, the following holds:
$$
\sqrt{k} \Vert x\Vert_2\geq \Vert x\Vert_1 \geq \delta \sqrt{k} \Vert x \Vert_2 
$$
A linear operator $A: \mathbb{R}^d \rightarrow \mathbb{R}^k$ is said to be $\delta$-Euclidean if the range of $A$ is a euclidean $(\delta,d,k)$  section.  We also note that because of Jensen's inequality, $\sqrt{k} \Vert x\Vert_2\geq \Vert x\Vert_1$ holds trivially. 
\end{definition}
We remark that when we say a subspace $A \subset \mathbb{R}^k$ is Euclidean, we simply mean that there is some constant $\delta>0$ such that $A$ is $(\delta,\dim(A),k)$ Euclidean. The next claim states a useful fact about Euclidean sections. We also introduce the following notation. For any $\ell_p$ norm defined on $\mathbb{R}^k$ and any $S \subseteq [k]$ and $x \in \mathbb{R}^k$, $\Vert x \Vert_{S,p}$ denotes the $\ell_p$ norm of the vector $x_S$ where $x_S$ is the projection of $x$ on the subset $S$. 
\begin{claim}\label{clm:euc}
Let $V \subseteq \mathbb{R}^k$ be a $(\delta, d , k)$ Euclidean section. Then for every $0 < \delta ' < \delta^2/4$, $S \subseteq [k]$ such  that $|S| \leq \delta' \cdot k$ and $x \in V$
$$
\Vert x \Vert _1  - 2 \cdot \Vert x \Vert_{S,1} \geq \beta\sqrt{k}  \Vert x \Vert_2
$$
where $\beta=\delta -2\sqrt{\delta'}$.
\end{claim}
\begin{proof}
Let $S \subseteq [k]$ and $|S| \leq \delta' \cdot k$. Then, by Jensen's inequality, we can say that 
$$
\Vert x \Vert_{S,1} \leq \sqrt{|S|} \cdot \Vert x \Vert_{S,2} \leq \sqrt{|S|} \cdot \Vert x \Vert_{2}
$$
This implies that 
$$
\Vert x \Vert_{1} - 2 \Vert x \Vert_{S,1} \geq \Vert x \Vert_{1}  - 2\sqrt{|S|} \cdot \Vert x \Vert_{2} \geq (\delta - 2\sqrt{\delta'}) \sqrt{k} \Vert x \Vert_2
$$
We get the stated result by putting $\beta = \delta - 2\sqrt{\delta'}$.
\end{proof}
\subsection*{Correctness of the LP decoding}
We next state a theorem which roughly says the following   : Assume $A$ is a linear map whose range is a Euclidean section and all the singular values of $A$ are at least $\sigma$. Also, let $F : (\mathbb{Z}^+)^d \rightarrow \mathbb{R}^k$ be the query corresponding to the linear map defined by $A$ (with $k=\Theta(d)$). Then, given noisy values of $A \cdot x$ where the noise is such that but for some (fixed) positive  fraction of the coordinates, all other coordinates are within $\alpha=o( \sigma)$ of the correct value, the linear program (\ref{eq:eqnLP}) gives $\tilde{x}$ such that $\Vert x - \tilde{x} \Vert_1 \ll d$. 
\begin{theorem}\label{thm:thmLPDe}
Let $A : \mathbb{R}^d \rightarrow \mathbb{R}^k$ be a full rank linear map ($k>d$) and all the singular values of $A$ are at least $\sigma$. Further, the range of $A$ (denoted by $\mathcal{L}(A)$) is a $(\delta, d ,k)$ Euclidean section. 
Let $F : (\mathbb{Z}^+)^d \rightarrow \mathbb{R}^k$ the query corresponding to $A$. Then, there exists $\gamma = \gamma(\delta)$ such that if
$$
\Pr_{i \in [k]} [|F(x)_i - \tilde{y}_i| \leq \alpha] \geq 1-\gamma
$$
then, any solution $\tilde{x}$ to the linear program (\ref{eq:eqnLP}) satisfies  $\Vert \tilde{x} - x \Vert_1 \le O(\alpha\sqrt{ kd}/\sigma)$ where the constant inside the $O(\cdot)$ notation depends on $\delta$. 
\end{theorem}
\begin{proof}
First of all, using Claim~\ref{clm:euc}, we get that for any $S \subseteq [k]$ such that $|S| \leq \delta^2k/8$ and $z \in \mathbb{R}^d$ 
\begin{equation}\label{eqn:eqnLPd}
\Vert A \cdot z \Vert_{1} - 2 \Vert A \cdot z \Vert_{S,1} \geq (\delta -2 \sqrt{\delta'}) \sqrt{k} \Vert A \cdot z \Vert_2 \geq \frac{\delta \sqrt{k}}{4} \Vert A \cdot  z \Vert_2 \geq  \frac{\delta \sigma \sqrt{k}}{4} \Vert   z \Vert_2  
\end{equation}
The last inequality uses that the singular value of $A$ is at least $\sigma$. Now, assume that $\tilde{x}$ is the solution to the linear program (\ref{eq:eqnLP}). Also, assume that $\tilde{y} = A \cdot x + e$.  Set $\gamma  = \delta^2/8$ and then given the assumption on the noise, $e$ satisfies
$$
\Pr_{i \in [k]} [|e_i| >\alpha] \leq \delta^2/8
$$
Also, let $\tilde{x} = x +z$. The next few steps are identical to \cite{DMT07} but we do it here for the sake of completeness.  Since, $\tilde{x}$ is the solution to the linear program, we get
$$
\Vert A \cdot x - \tilde{y} \Vert_1\geq \Vert A \cdot \tilde{x} - \tilde{y} \Vert_1  
$$
$$
\Vert A \cdot x - \tilde{y} \Vert_1 \geq \Vert A \cdot (\tilde{x} - x) + A \cdot x  - \tilde{y} \Vert_1 
$$
$$
\Vert e \Vert_1 \geq \Vert A \cdot z - e\Vert_1 
$$

Let $S   = \{ i : |e_i | > \alpha \}$. Then, from the above, we get that 
$$
\Vert e \Vert_{S,1} + \Vert e \Vert_{\overline{S},1} \geq \Vert A \cdot z -  e \Vert_{S,1} + \Vert A \cdot z -  e  \Vert_{\overline{S},1} \geq \Vert e \Vert_{S,1}  -  \Vert A \cdot z \Vert_{S,1} + \Vert A \cdot z -  e  \Vert_{\overline{S},1}  
$$ 
We next use $ \Vert A \cdot z -  e  \Vert_{\overline{S},1}  \geq \Vert A \cdot z \Vert_{\overline{S},1}  -  \Vert e  \Vert_{\overline{S},1} $ on the right hand side of the above inequality to simplify and get 
$$
 \Vert e \Vert_{\overline{S},1} \geq  \Vert A \cdot z   \Vert_{\overline{S},1}   - \Vert e   \Vert_{\overline{S},1}   -  \Vert A \cdot z \Vert_{S,1} 
$$
$$
   \Vert A \cdot z   \Vert_{1} - 2 \cdot \Vert A \cdot z \Vert_{S,1} \leq 2  \Vert e   \Vert_{\overline{S},1} 
$$
$$
   \Vert A \cdot z   \Vert_{1} - 2 \cdot \Vert A \cdot z \Vert_{S,1} \leq 2 \alpha   |\overline{S}|
$$
$$
\Vert A \cdot z \Vert_1 - 2 \Vert A \cdot z \Vert_{S,1} \leq 2 \alpha (1-\delta^2/8) k
$$
Combining the above with (\ref{eqn:eqnLPd}), we get that 
$$
\Vert z \Vert_2 \leq \frac{8 \alpha (1-\delta^2/8) \sqrt{k}}{\delta \sigma} \quad \Rightarrow \quad \Vert z \Vert_1 \leq \frac{8 \alpha (1-\delta^2/8)\sqrt{k d} }{\delta \sigma} = O\left( \frac{\alpha \sqrt{kd}}{\sigma} \right)
$$
\end{proof}
\subsection{Releasing $\ell$-way marginals privately}
We now formally define the problem of releasing $\ell$-way marginals.  The universe $X$ is of the form $\{-1,1\}^{d'}$ and we have a counting query $F:( \mathbb{Z}^+)^d \rightarrow (\mathbb{R})^k$ where $k = \binom{d'}{\ell} \cdot 2^{\ell}$ and $d=2^{d'}$. For any $x \in ( \mathbb{Z}^+)^d$, each coordinate of $x$ is identified with an element of $\{-1,1\}^{d'}$. Note that $[d']$ can be identified with a set of $d'$ binary attributes and any $z \in \{-1,1\}^{d'}$ denotes a specific setting of the attributes. Hence, for $x \in ( \mathbb{Z}^+)^d$, the row associated with $z \in \{-1,1\}^{d'}$ counts how many elements have that specific setting of  the $d'$ attributes.

 Similarly, consider any $c \in \{-1,0,1\}^{d'}$ and define $|c|  = \sum_{j=1}^{d'} |c_i|$ \ie $|c|$ represents the number of non-zero entries in $c$.  We note that the set $S$ defined as 
$$
S = \{ c : c \in \{-1,0,1\}^{d'} \textrm{ and } |c| = \ell \}
$$
has size $2^{\ell} \cdot \binom{d'}{\ell} = k$. Thus, every $c \in S$  can be identified with some specific setting of a specific set of $\ell$ among the $d'$ attributes. Hence, we identify the set $S$ with $[k]$ and for any $y \in \mathbb{R}^k$, every coordinate of $y$ is identified with some specific $c \in S$. We now define $F$. Consider a database $x$ and a row corresponding to $c$. Let $Z$ be the set of $\ell$ attributes and $\Delta$ be their setting corresponding to $c$. Then the row of $F(x)$ corresponding to $c$ simply counts the number of elements in $x$ where the attributes in $Z$ are indeed set to $\Delta$.  We now define $F$ more formally. 

Now, for any $c \in \{-1,0,1\}^{d'}$ and $z \in \{-1,1\}^{d'}$, we define $F_c(z)$ as follows :
$$
F_c(z) = \prod_{j : c_j=1} \left( \frac{1+z_j}{2} \right) \prod_{j : c_j=-1} \left( \frac{1-z_j}{2} \right)
$$
In other words, $F_c(z)$ is $1$ iff the following holds for every $j$ : If $c_j=1$, then $z_j=1$ and if $c_j=-1$, then $z_j=-1$. 

We now define the counting query $F:( \mathbb{Z}^+)^d \rightarrow (\mathbb{R})^k$ as
$$
F(x) = \left(\sum_{z \in \{-1,1\}^{d'}} F_1(z), \sum_{z \in  \{-1,1\}^{d'}} F_2(z),\ldots,  \sum_{z \in  \{-1,1\}^{d'}} F_k(z)\right)
$$
We now state our main theorem of this section.
 \begin{theorem}
 Let $q,\ell \in \mathbb{N}$ be constant integers. Then, there exists a constant $\gamma = \gamma(q,\ell)>0$  such that any mechanism which releases the $\ell$-way marginals of a table of size $n$ over $d'$ attributes and $ n \le d'^{\ell-1} \log_{(q)} n$ by adding at most $\eta$ noise to $1-\gamma$ fraction of the queries where 
 $$
 \eta = o(\sqrt{n})
 $$
is attribute non-private. Further, the algorithm which violates attribute privacy is efficient and uses LP decoding. 
 \end{theorem}
Here $ \log_{(q)} n$ is an  iterated logarithm which is defined precisely later on. However, we wanted to state the main result of this section in the beginning itself before diving into the proof structure.

In order to go into the proof structure, we first  state the following theorem due to Kasiviswanathan \etal \cite{KRSU10} which is a weaker version of our result. We state their result only for the case when $d'^{\ell-1} \gg n \cdot \log^{2\ell-4} n$  because our subsequent improvements are valid (and possible) only in the regime of $d'^{\ell-1} \gg n$. 
\begin{theorem}\label{thm:thmKRSU}\cite{KRSU10}
Let $\ell \in \mathbb{N}$ be a constant and $n,d \in \mathbb{N}$ such that  $d'^{\ell-1} \gg n \cdot \log^{2\ell-4} n$. Then, for every mechanism $M$ which releases $\ell$-way marginals  of a database of size $n$ (and universe $\B^{d'}$) such that the noise for every single query is bounded by $\eta$ where
$$
\eta \ll \frac{\sqrt{n}}{\log^{\ell^2- \ell +1} n}
$$
is attribute non-private. In fact, the attack is an efficient algorithm based on $\ell_2$ norm minimization. 
\end{theorem}
Before, we glimpse into how they prove their result and our improvement on that, we need to describe the Hadamard product of matrices. 
\begin{definition}
Let $A_1, \ldots, A_{s} \in \mathbb{R}^{\ell_i \times n}$. Then, the Hadamard product of $A_1, \ldots, A_s$ is denoted by $A =A_1 \circ A_2 \ldots \circ A_s \in \mathbb{R}^{L \times n}$ where $L = \ell_1 \cdot \ldots \cdot \ell_s$ and is defined as follows :  Every row of $A$ is identified with a unique element of $[\ell_1] \times \ldots \times [\ell_s]$. Also corresponding to $i= (i_1,\ldots, i_s)$, 
$$
A[i,k] = \prod_{j=1}^s A_j[i_j,k]
$$
where $A[i,k]$ represents the element in row $i$ and column $k$. 
\end{definition}
The attack in \cite{KRSU10} is an efficient algorithm (is simply matrix inversion)   and is basically dependent on showing existence of a Hadamard product of small matrices such that all its singular values are large. In particular, they show the following reduction. 
\begin{lemma}\label{lem:lemKRSU}
Let $A_1  , \ldots, A_{\ell-1} \in \mathbb{R}^{d' \times n}$ such that  $A =A_1 \circ A_2 \ldots \circ A_{\ell-1}$ (with $d'^{\ell-1} >n$)  such that all singular values of $A$ are at least $\sigma$. Further, all entries of $A_1, \ldots, A_{\ell-1}$ are in the set $\B$. Then, any mechanism $M$ which adds $\ll (\sqrt{n} \cdot \sigma)/(\sqrt{d'^{\ell-1}})$ noise to every coordinate while releasing all the $\ell$-way marginals of a table of size $n$ with $d'$ attributes, is attribute non-private. 
\end{lemma}
Subsequently, to prove Theorem~\ref{thm:thmKRSU}, they proved the following lemma. 
\begin{lemma}
Let $D \sim \mathbb{R}^{d' \times n}$ be a distribution over matrices such that every entry of the matrix is  an independent and unbiased $\B$ random variable.  Let $A_1, \ldots, A_{\ell-1}$ be i.i.d. copies of random matrices drawn from the distribution $D$ and $A$ be the Hadamard product of $A_1, \ldots, A_{\ell-1}$. Then, provided that  $d'^{\ell-1} \gg n \cdot \log^{2\ell-4} n$, with probability $1-o(1)$, the smallest singular value of $A$ denoted by $\sigma_{n} (A)$ satisfies
$$
\sigma_n(A) \geq  \frac{\sqrt{d'^{\ell-1}}}{\log^{\ell^2- \ell +1} n}
$$
\end{lemma} 
This lemma clearly sufficed to prove Theorem~\ref{thm:thmKRSU}. One of the main avenues of improvement in the result in \cite{KRSU10} is that because the attack is matrix inversion (or in other words, $\ell_2$ distance minimization), it is inherently a noise sensitive method. In other words, it is not evident from the result in \cite{KRSU10} that even if $1$ of the $\ell$-way marginals is answered with arbitrary noise, if there is still an attack which can violate attribute privacy. 
In this respect, LP decoding seems like a much more powerful method and is usually robust even if a fraction of the queries are answered with wild noise. In particular, we can combine the technique in proof of Theorem~\ref{thm:thmLPDe} and Lemma~\ref{lem:lemKRSU} to get the following lemma (We do not prove it here because it is a straightforward combination of the two techniques). 
\begin{lemma}\label{lem:aux}
Let $A_1  , \ldots, A_{\ell-1} \in \mathbb{R}^{d' \times n}$ such that  $A =A_1 \circ A_2 \ldots \circ A_{\ell-1}$ (with $d'^{\ell-1} >n$)  and all entries of $A_1, \ldots, A_{\ell-1}$ are in the set $\B$. Also, all the singular values of $A$ are at least $\sigma$ and the range of $A$ \ie $\mathcal{L}(A)$ is a $(\delta, n , d'^{\ell-1})$ Euclidean section. Then, there exists a constant $\gamma= \gamma(\delta) >0$ such that any mechanism which answers at least $1-\gamma$ fraction of the $\ell$-way marginals with noise bounded by $\alpha$ is attribute non-private provided $\frac{ \alpha \sqrt{d^{\ell'-1}  \cdot n}}{ \sigma} =o(n)$ or in other words, $\alpha = o(\sqrt{n}\sigma/\sqrt{d^{\ell'-1}})$
\end{lemma}
To describe the main technical result of Rudelson \cite{Rud11}, we need to define iterated logarithms. 
\begin{definition}
For $r \in \mathbb{N}$, we define $\log_{(r)}n $ as follows :
\begin{itemize}
\item $\log_{(1)}n = \max \{ \log_2 n,1 \} $ 
\item $\log_{(r)}n = \log_{(1)} \ (\log_{(r-1)} n)$
\end{itemize}
\end{definition}
\begin{theorem}\label{thm:Rud11}\cite{Rud11}
Let $q,\ell \in \mathbb{N}$ be constants. Also, let  $D \sim \mathbb{R}^{d' \times n}$ be a distribution over matrices such that every entry of the matrix is  an independent and unbiased $\B$ random variable. 
 Let $A_1, \ldots, A_{\ell-1}$ be i.i.d. copies of random matrices drawn from the distribution $D$ and $A$ be the Hadamard product of $A_1, \ldots, A_{\ell-1}$. Then, provided that  $d'^{\ell-1} \gg n \log_{(q)} n$, with probability $1-o(1)$, the smallest singular value of $A$ denoted by $\sigma_{n} (A)$ satisfies
$$
\sigma_n(A) = \Omega(\sqrt{d'^{\ell-1}})
$$
Also, the range of $A$ is a $(n,d'^{\ell-1},\gamma(q,\ell))$ Euclidean section for some $\gamma(q,\ell)>0$. 
 \end{theorem}
 Theorem~\ref{thm:Rud11} and Lemma~\ref{lem:aux} immediate imply our main theorem which we restate here for convenience. 
 \begin{theorem}
 Let $q,\ell \in \mathbb{N}$ be constant integers. Then, there exists a constant $\gamma = \gamma(q,\ell)>0$  such that any mechanism which releases the $\ell$-way marginals of a table of size $n$ over $d'$ attributes and $ n \le d'^{\ell-1} \log_{(q)} n$ by adding at most $\eta$ noise to $1-\gamma$ fraction of the queries where 
 $$
 \eta = o(\sqrt{n})
 $$
is attribute non-private. Further, the algorithm which violates attribute privacy is efficient and uses LP decoding. 
 \end{theorem}

\section{Noise lower bounds for blatant non-privacy}
In this section, we prove lower bounds on the noise required to prevent blatant non-privacy while answering random counting queries. Dinur and Nissim \cite{DN03}, in their seminal paper, had shown that answering $O(n \log^2 n)$ subset sum queries  with $o(\sqrt{n})$ noise results in blatant non-privacy. In other words, they had proven the following theorem. 
\begin{theorem}\cite{DN03}
For every $n \in \mathbb{N}$, there is a counting query $F :\B^n \rightarrow \mathbb{R}^k$ for $k=O(n \log^2n)$ and an algorithm $A$ such that if $M(x,F)$ is the answer returned by the mechanism for the database $x$ and query $F$ and
$$
\forall i \  \ |M(x,F)_i - F(x)_i| = o(\sqrt{n})
$$
then
$$
A(M(x,F)) =x ' : \Vert x - x' \Vert_1 = o(n)
$$
\end{theorem}
The algorithm $A$ in the above result is efficient and uses linear programming. Since then, several improvements were made to this result including the results in \cite{DMT07} where the same conclusion was achieved under the weaker hypothesis that 
$$
\Pr_{ i \in [k]} [ |M(x,F)_i - F(x)_i| = o(\sqrt{n}) ] \ge 1/2 +\eta
$$
for any $\eta>0$. While the attack was inefficient, they also showed how to use LP decoding to get an efficient attack and they could achieve this under the  weaker hypothesis 
$$
\Pr_{ i \in [k]} [ |M(x,F)_i - F(x)_i| = o(\sqrt{n}) ] \ge 0.761$$
Further, in the results in \cite{DMT07}, the domain of the database was $\mathbb{R}^n$ as opposed to $\{0,1\}^n$. 
However, all these results achieved blatant non-privacy when the size of the universe was same as the size of the database namely $n$. In this section, we consider the modified setting where the size of the universe is smaller than size of the database. As the next theorem shows, when the universe size is smaller, the noise required to prevent blatant non-privacy is much larger.

\begin{theorem}
For every $n,d \in \mathbb{N}$ and $\eta>0$, there is an algorithm $A$ and   a counting query $F : (\mathbb{Z}^+)^d \rightarrow \mathbb{R}^k$ with $k=8(d \cdot \log n) /\eta^2$ such that if $M(x,F)$ is the answer returned by the mechanism and
$$
\Pr_i [|M(x,F)_i - F(x)_i| = o\left( \frac{\dd}{\sqrt{d}} \right)] \ge 1/2+ \eta
$$
then, 
$$
A(M(x,F)) =x ' : \Vert x - x' \Vert_1 = o(\dd)
$$
Similarly, for every $n,d \in \mathbb{N}$ and $\theta =\theta(n,d) \in \mathbb{R}^{+} \le n$, there is an algorithm $B$ and   a counting query $F : (\mathbb{Z}^+)^d \rightarrow \mathbb{R}^k$ with $k= 2d \log n \cdot  \exp \left( \frac{2d \theta^2}{\dd^2}\right) $ such that if $M(x,F)$ is the answer returned by the mechanism and
$$
\forall i  \quad  [|M(x,F)_i - F(x)_i| = o(\theta)] 
$$
then, 
$$
B(M(x,F)) =x ' : \Vert x - x' \Vert_1 = o(\dd)
$$
\end{theorem}
Before going ahead with the proof, we remark that while both the algorithms $A$ and $B$, as described in the proof are inefficient, the former can be made efficient using linear programming (along the lines of \cite{DN03,DMT07}). We choose not to do it for the sake if simplicity. 
\begin{proof}
The algorithm $A$ is as follows : 
\begin{itemize}
\item Enumerate over all $x' \in (\mathbb{Z}^+)^d$ such that $\Vert x \Vert_1 \le n$. 
\item Return $x'$ if $\Pr_i [|M(x',F)_i - F(x)_i| = o\left( \frac{\dd}{\sqrt{d}} \right)] \ge 1/2+ \eta/4$
\end{itemize}
It is clear that the algorithm returns an answer as $x'=x$ satisfies the second condition. Now, consider any $c>0$ and $x'$ and $x$ which differ by $c\dd$.  Then, choose $r_1, \ldots, r_k \in \{-1,1\}^{n}$ uniformly at random and let us define the linear query $F: (\mathbb{Z}^+)^d \rightarrow \mathbb{R}^k$   as follows :
$$
F(x) = (\sum_{j=1}^d x_j \cdot r_{1j}, \ldots, \sum_{j=1}^d x_j \cdot r_{kj} ) 
$$
Now, consider $F_\ell (x) - F_{\ell}(x') = \sum_{j=1}^d (x_j-x'_j) \cdot r_{\ell j} = \sum_{j=1}^d z_j \cdot r_{\ell k}$ where $z= x -x'$ with $\Vert z \Vert_1 \geq c\dd$. Then, by Theorem~\ref{thm:thmrad}, $\exists  c'= c'(c,\eta)$ such that 
$$
\Pr_{r_{\ell} \in \{-1,1\}^d} \left[|\sum_{j=1}^d (x_j-x'_j) \cdot r_{\ell j}| \leq \frac{c'\dd}{\sqrt{d}}  \right] \leq \eta/100
$$
Let us define the set $S = \left\{ \ell :  |\sum_{j=1}^d (x_j-x'_j) \cdot r_{\ell j}| \geq \frac{c'\dd}{\sqrt{d}}  \right\}$. Then a Chernoff bound implies that 
\begin{equation}\label{eqn:union}
\Pr_{r_1, \ldots, r_k} [|S| 
> (\eta \cdot k)/4 ] \leq 1-\exp(-\eta^2 k/4) 
\end{equation}
Putting $ k = \frac{8}{\eta^2} \cdot d \cdot \log n$, we get that the probability is strictly bounded by $n^{-2d}$. Note that the total number of pairs $x,x' \in (\mathbb{Z}^+)^d$ such that $\Vert x \Vert_1, \Vert x' \Vert_1 \leq n$ is bounded by $n^{2d}$. Thus applying the union bound with (\ref{eqn:union}), we get that we can fix $r_1,\ldots, r_k$ and get $F: (\mathbb{Z}^+)^d \rightarrow \mathbb{R}^k$ such that  if $x,x' \in (\mathbb{Z}^+)^d$, $\Vert x \Vert_1\leq n$, $\Vert x' \Vert_1 \leq n$ and $\Vert x-x' \Vert_1 \geq c\dd$, then
$$
\Pr_{i \in [k]} \left[|F(x)_i - F(x')_i| \geq \frac{(c'\dd)}{\sqrt{d}}\right] > 1-\eta/4
$$
Using the hypothesis \ie 
$$
\Pr_i \left[|M(x,F)_i - F(x)_i| = o\left( \frac{\dd}{\sqrt{d}} \right)\right] \ge 1/2+ \eta
$$
we get that for all $x'$ such that $\Vert x -x' \Vert_1 \geq cn$, 
$$
\Pr_i \left[|M(x,F)_i - F(x')_i| = o\left( \frac{\dd}{\sqrt{d}} \right)\right] < 1/2+ \eta/4
$$
Thus, when we have the query $F$ described above,  the algorithm $A$ will never return $x'$ if $\Vert x-x' \Vert_1 \geq c\dd$ and hence the correctness is proven. 

We now come to the second part of the theorem. The algorithm $B$ is described as follows : 
\begin{itemize}
\item Enumerate over all $x' \in (\mathbb{Z}^+)^d$ such that $\Vert x \Vert_1 \le n$. 
\item Return $x'$ if $\forall i$, $ [|M(x',F)_i - F(x)_i| = o(\theta)$
\end{itemize}
Again, first of all, we observe the algorithm $B$ always returns an answer as $x'=x$ satisfies the second condition. 
To prove the second part of the  result, we define $F :  (\mathbb{Z}^+)^d \rightarrow \mathbb{R}^k$ exactly in the same way that we did earlier. Now, consider any $\ell \in [k]$. Now, consider $F_\ell (x) - F_{\ell}(x') = \sum_{j=1}^d (x_j-x'_j) \cdot r_{\ell j} = \sum_{j=1}^d z_j \cdot r_{\ell k}$ where $z= x -x'$ with $\Vert z \Vert_1 \geq c\dd$. Then, by Theorem~\ref{thm:thmrad},
$$
\Pr_{r_{\ell} \in \{-1,1\}^d} \left[|\sum_{j=1}^d (x_j-x'_j) \cdot r_{\ell j}| \geq c \cdot \theta  \right] \ge \exp \left( \frac{2d \theta^2}{\dd^2}\right) 
$$
This means that putting $k =  \exp \left( \frac{2d \theta^2}{\dd^2}\right)  \cdot 2d \log n$
$$
\Pr_{r_1,\ldots, r_{k} \in \{-1,1\}^d} \left[ \exists \ell \in[k] : |\sum_{j=1}^d (x_j-x'_j) \cdot r_{\ell j}| \geq c \cdot \theta  \right] >1- n^{-2d}$$
As before, by a union bound, we can fix $r_1, \ldots, r_k$ and get a query $F: (\mathbb{Z}^+)^d \rightarrow \mathbb{R}^k$ such that for any $x, x' \in (\mathbb{Z}^+)^d$, $\Vert x \Vert_1 \le n$, $\Vert x' \Vert_1 \le n$ and $\Vert x - x' \Vert_1 \geq c\dd$, 
$$
\exists i \in [k] : \quad |F(x)_i - F(x')_i | \geq c \cdot \theta 
$$
Using the hypothesis \ie
$$
\forall i  \quad  [|M(x,F)_i - F(x)_i| = o(\theta)] 
$$
Hence, we get that (for all $x'$ such that $\Vert x -x' \Vert_1 \geq c\dd$)
$$
\exists i \in [k] : \quad |M(x,F)_i - F(x')_i | = \Omega(\theta)
$$
This shows that algorithm $B$ will never return $x'$ if $\Vert x-x' \Vert_1 \geq c\dd$ and hence the correctness is proven. 
\end{proof}
\section*{Acknowledgements} 
First and foremost, I would like to  thank  Cynthia Dwork for introducing me to the problems discussed in this paper,  immense help with the presentation and the technical help. Even though she declined to co-author the paper, without her contributions, this paper would not have existed.  This work was almost entirely done during a very enjoyable summer in 2010 at MSR Silicon Valley while the author was a summer intern with her.  I would also like to thank Salil Vadhan for his kind permission to include the results of subsection~\ref{sec:infoloss} in this paper. 

The author would like to thank Moritz Hardt and Mark Rudelson for very helpful conversations. The question of getting a lower bound for blatant non-privacy with dependence on universe size came up in a discussion with Moritz Hardt.  I would like to thank Mark for answering countlessly many questions about random matrices and anti-concentration. I also had useful conversations about this work with Ilya Mironov, Elchanan Mossel, Omer Reingold, Adam Smith, Alexandre Stauffer, Kunal Talwar, and Salil Vadhan. 

I would also like to thank the SODA 2012 and TCC 2012 reviewers for many useful comments including pointing out an error in the earlier proof of Lemma~\ref{lem:mutual}.

\bibliography{macros,luca}
\appendix
\section{Construction of databases with large differences}
In this section, we give constructions of sets of vectors in $(\mathbb{Z}^+)^d$ (in other words, sets of databases) which were used in the results in Section~\ref{sec:diffpriv}.
\begin{claim}\label{clm:clm0}
There exists a constant $C>1$ such that for any $d \in \mathbb{N}$, $0<\epsilon<1$  and $\log d \le k \le d/C$, there is an integer $a \ge (\log(d/k)/160\epsilon)$ such that  it is possible to construct $2^{k/20}$ vectors $x_1, \ldots, x_{2^{k/20}} \in \mathbb{Z^{+}}^{d}$ with the following properties :
\begin{itemize}
\item Every entry of $x_i$ is either $0$ or $a$ 
\item $\forall i$, $\Vert x_i \Vert_1 \leq k/80\epsilon$ and $\Vert x_i - x_j \Vert_1\geq k/160\epsilon$
\end{itemize}
\end{claim}
\begin{proof}
We use construction of combinatorial designs from \cite{EFF85,RRV99}. Namely, the main theorem in \cite{EFF85} states that it is possible to construct sets $S_1, \ldots, S_m \subseteq [d]$ with the following properties : 
\begin{itemize}
\item $\forall i$, $|S_i| = \ell$
\item  $\forall i \not =j$, $|S_i \cap S_j | \leq \rho$ 
\end{itemize}
provided $d \geq  C' \cdot \frac{\ell^2 \cdot m^{1/\rho}}{\rho}$ for some large constant $C'$. We now see that if we put $\ell=\lfloor k/\log(d/k) \rfloor $, $\rho = \ell/4$ and $m=2^{k/20}$, then indeed $d \geq  C' \cdot \frac{\ell^2 \cdot m^{1/\rho}}{\rho}$ provided $C$ is sufficiently large compared to $C'$. 

Now, consider the set $A = \{y_1, \ldots, y_{2^{k/20}} \}$ be the characteristic vectors of the sets $S_1, \ldots, S_{2^{k/20}}$.   Now, we observe that setting $x_i = \lfloor (\log(d/k)/80\epsilon) \rfloor\cdot y_i$ achieves all the stated conditions. 
\end{proof}
\begin{claim}\label{clm:clm1}
For any $d,d' \in \mathbb{N}$, $d' \le d$ and $1/40 \geq \epsilon>0$, $n=d'/(4\epsilon)$, there exists $S \subseteq \mathbb{Z^{+}}^{d}$ such that the following conditions hold:
\begin{itemize}
\item $\forall x \in S $, $\Vert x \Vert_1 \le n$
\item $\forall x, y \in S$ and $x \not =y$, $\Vert x-y \Vert_1 \ge n/10$
\item $|S| \ge 2^{4d'/5}$
\end{itemize}
\end{claim}
\begin{proof}
Consider a set $C \subseteq \{0,1\}^d$ with the following two properties. 
\begin{itemize}
\item $|C| \geq 2^{4d'/5}$
\item $\forall x, y \in C$, $x \not = y$, $\Vert x-y\Vert_1 \ge d'/9$
\item $\forall x \in C$, $\Vert x \Vert_1 \leq d'$. 
\end{itemize}
Such a set $C$ exists. To see this, consider an error correction code $C' \subseteq \B^{d'}$ with distance $d'/9$ and rate $4/5$. Such a code exists via probabilistic method. Now, the set $C$ is constructed as
$$
C = \{ x \equiv y \circ 0^{d-d'} : y \in C' \}
$$
  Now, consider the set $S \in \mathbb{Z^{+}}^{d}$ such that 
$$
S = \{x : x = \lfloor(1/4\epsilon)\rfloor\cdot z \textrm { and } z \in C \}
$$
We claim that the set $S$ satisfies the required conditions. The first and the third parts of claim are obvious. Note that for any $x , y \in S$, $\Vert x - y \Vert_1 \geq (d'/9) \cdot  \lfloor(1/4\epsilon)\rfloor \geq d'/(40 \epsilon) \geq n/10 $. The penultimate inequality uses that $(1/4\epsilon) \ge 10$. 
\end{proof}
The above claim worked for $\epsilon<1$. We now prove a claim which works in the regime of $\epsilon>1$. 
\begin{claim}\label{clm:clm6}
There exists $C>0$ such that for any $d,d' \in \mathbb{N}$, $\epsilon>1$ and $d' \leq d/(C \cdot 2^{5\epsilon})$ , $n=d'/(4 \epsilon)$, there exists $S \subseteq \mathbb{Z^{+}}^{d}$ such that the following conditions hold:
\begin{itemize}
\item $\forall x \in S $, $\Vert x \Vert_1 \le  n  $
\item $\forall x, y \in S$ and $x \not =y$, $\Vert x-y \Vert_1 \ge n/10$
\item $|S| \ge 2^{4d'/5}$
\end{itemize}
Further, $S$ is in fact a subset of $\{0,1\}^d$. 
\end{claim}
\begin{proof}
We observe that the construction of set $S$ is related to the construction of combinatorial designs \cite{EFF85,RRV99}  with specific parameters. In particular, the result in \cite{EFF85} allows us to construct sets $S_1, \ldots, S_m \subseteq [d]$ with the following properties : 
\begin{itemize}
\item $\forall i$, $|S_i| \le \lfloor n \rfloor$
\item $\forall i \not = j$, $|S_i \cap S_j| \leq \rho \leq \lfloor 4n/5 \rfloor$ 
\item $m \geq 2^{4d'/5}$ 
\end{itemize}
provided that  $d \ge \frac{C'\cdot n^2 \cdot m^{1/ \rho}}{   \rho }$ (for some large constant $C'$).  Using the conditions on $d,d'$ and $n$, we see that the condition is satisfied provided $C$ is sufficiently large compared to $C'$.  Clearly, if $x_1, \ldots, x_m$ are characteristic vectors of the sets $S_1,\ldots, S_m$ respectively, then $$
S =\{x_1,\ldots, x_m \}$$ 
satisfies the conditions of our claim. 
\end{proof}
\section{Large deviation of Rademacher sums from their mean}
In this section, we prove the following inequality which says that Rademacher sums have large deviations from their mean with significant probability. 
\begin{theorem}\label{thm:thmrad}
Let $x_1,\ldots,x_d$ be i.i.d. $\pm 1$ random variables such that $\mathbb{E} [x_i]=0$. Also, let $a_1, \ldots, a_d \in \mathbb{R}^{+}$ such that $\sum_{i=1}^d a_i = n$.  Then, for any $0\le \theta \leq n/2$, 
$$
\Pr_{x_1,\ldots, x_d} \left[ \left|\sum_{i=1}^d  a_i x_i \right| > \theta\right]  \geq \exp\left(\frac{-2d\theta^2}{n^2}\right)
$$
\end{theorem}
An immediate application of the above theorem is the following corollary. 
\begin{corollary}\label{cor:cor1}
Let $x_1,\ldots,x_d$ be i.i.d. $\pm 1$ random variables such that $\mathbb{E} [x_i]=0$. Also, let $a_1, \ldots, a_d \in \mathbb{R}^{+}$ such that $\sum_{i=1}^n a_i = n$. Then,
$$
\Pr_{x_1,\ldots, x_d} \left[ \left|\sum_{i=1}^d  a_i x_i \right| > \frac{n}{10 \sqrt{d}}\right] \geq \frac{9}{10}
$$
\end{corollary} 
To prove Theorem~\ref{thm:thmrad}, we use 
 the following result due to Montgomery-Smith \cite{Mon90}. For $y \in \mathbb{R}^n$ and $t>0$, we define
$$
K(y,t)  = \inf _{y' \in \mathbb{R}^n} \{   \Vert y-y' \Vert_1 + t \Vert y'\Vert_2           \}
$$
The following theorem about large deviation of Rademacher sums from the mean was proven by Montgomery-Smith \cite{Mon90}. 
\begin{theorem}\label{lem:MS}
Let $x_1,\ldots,x_d$ be i.i.d. $\pm 1$ random variables such that $\mathbb{E} [x_i]=0$.  Then, 
$$
\Pr_{x_1,\ldots, x_d} \left[ \left|\sum_{i=1}^n  y_i x_i \right| > K(y,t)\right] \geq \exp(-t^2/2)
$$
\end{theorem}
To use Theorem~\ref{lem:MS}  in order to prove the Theorem~\ref{thm:thmrad}, we make the following claim. 
\begin{claim}
Let $ y \in \mathbb{R}^d$ such that $\Vert y\Vert_1=n$. Then, $ K(y,t)  > \min \{n/2, nt /(2\sqrt{d}) \}$.   
\end{claim}
\begin{proof}
Note that $ K(y,t)  = \inf _{y' \in \mathbb{R}^n} \{   \Vert y-y'\Vert_1 + t \Vert y'\Vert_2           \}$. Consider $y'$ which achieves the infimum. If $\Vert y' \Vert_2 \ge n /(2\sqrt{d})$, then we are done. Else, we get that $\Vert y'\Vert_1 \le n/2$ (by Cauchy - Schwarz). Then by triangle inequality, we get that $\Vert y-y' \Vert _1 \geq n/2$. 
\end{proof}
The theorem immediately follows by plugging the lower bound on $K(y,t)$ from the above claim in Theorem~\ref{lem:MS}.
\section{Concentration of measure for the sum of squares of Gaussians}\label{sec:chi}
In this section, we prove a result about the concentration of measure for the sum of squares of i.i.d. $\mcal{N}(0,\sigma)$ random variables. While this seems to be a well studied distribution in literature, 
we could not find a usable result on its concentration and hence we prove the following theorem here. 
\begin{theorem}\label{thm:chi}
Let $X_1, \ldots, X_k$ be $k$ i.i.d. $\mcal{N}(0,\sigma)$ random variables. Then, 
$$
\Pr_{X_1, \ldots, X_k} [ X_1^2 + X_2^2 + \ldots + X_k^2 >2 (1+\eta)k \sigma^2] \le 2^{\frac{-\eta k}{2}}
$$
\begin{proof}
Note that by definition, for any $i \in [k]$,  $$ \Pr[X_i =t] = \sqrt{\frac{1}{2 \pi \sigma^2}} \ e^{-\frac{t^2}{2 \sigma^2}} \ dt$$
Then, consider the random variable $Z_i = \exp\left(\frac{X_i^2}{4 \sigma^2}\right)$. We note that 
$$
\mathbb{E}[Z_i]  = \int_{-\infty}^{\infty} \sqrt{\frac{1}{2 \pi \sigma^2}} \ e^{-\frac{t^2}{2 \sigma^2}} \cdot e^{\frac{t^2}{4 \sigma^2}} \ dt =\int_{-\infty}^{\infty} \sqrt{\frac{1}{2 \pi \sigma^2}} \ e^{-\frac{t^2}{4 \sigma^2}} \ dt =\sqrt{2}
$$
Now, observe that 
\begin{eqnarray*}\Pr_{X_1, \ldots, X_k} [ X_1^2 + X_2^2 + \ldots + X_k^2 > \lambda ] &=& \Pr_{X_1, \ldots, X_k} \left[ \frac{X_1^2 + X_2^2 + \ldots + X_k^2}{4\sigma^2} > \frac{\lambda}{4\sigma^2} \right] \\ 
&=& \Pr_{X_1, \ldots, X_k} \left[ \exp\left(\frac{X_1^2 + X_2^2 + \ldots + X_k^2}{4\sigma^2}\right) > \exp\left(\frac{\lambda}{4\sigma^2} \right)\right]  \\ 
& \le & \left(\mathop{\mathbb{E}}_{X_1, \ldots, X_k}    \exp\left(\frac{X_1^2 + X_2^2 + \ldots + X_k^2}{4\sigma^2}\right)\right)   /      \exp\left(\frac{\lambda}{4\sigma^2} \right)     \end{eqnarray*} 
Using independence of the $X_i$'s, we get that the above expression is 
$$
\frac{\prod_{i=1}^k  \left(\mathop{\mathbb{E}}\exp\left(\frac{X_i^2 }{4\sigma^2}\right)\right) }{  \exp\left(\frac{\lambda}{4\sigma^2} \right) }   = \frac{\prod_{i=1}^k \mathbb{E}  [Z_i] }{\exp\left(\frac{\lambda}{4\sigma^2} \right) } =\frac{2^{k/2}}{\exp\left(\frac{\lambda}{4\sigma^2} \right) }  
$$
Putting $\lambda = 2 (1+\eta)k \sigma^2$, we get that the above expression is 
$$
\frac{2^{k/2}}{\exp\left(\frac{\lambda}{4\sigma^2} \right) }   = \frac{2^{k/2}}{ \exp\left(\frac{2 (1+\eta)k \sigma^2}{4\sigma^2} \right) } \le  \frac{2^{k/2}}{ 2^{\frac{ (1+\eta)k }{2} } } \leq 2^{-\frac{\eta k}{2}}
$$
\end{proof}
\end{theorem}

\end{document}